\documentclass[12pt,a4paper]{amsart}
\pdfoutput=1
\usepackage{graphicx}
\usepackage{stmaryrd}
\usepackage{amssymb}
\usepackage{amsmath}
\usepackage{amsthm}
\usepackage{dsfont}
\usepackage{hyperref}
\usepackage{amsfonts}
\usepackage{amsrefs}
\usepackage{epsfig}
\usepackage{latexsym}
\usepackage{cite}

\usepackage{amsthm}
\usepackage{enumerate}
\usepackage{centernot}
\usepackage{fullpage}
\usepackage{mathrsfs}

\usepackage{color}

\newtheorem{theorem}{Theorem}[section]

\theoremstyle{definition}

\numberwithin{equation}{section}

\begin{document}
\title{On the density of 2D critical percolation gaskets and anchored clusters}

\author{Federico Camia}

\address{Division of Science, NYU Abu Dhabi, Saadiyat Island, Abu Dhabi, UAE \& Courant Institute of Mathematical Sciences, New York University,
	251 Mercer st, New York, NY 10012, USA. \\ ORCiD: 0000-0002-3510-8535}
\email{federico.camia@nyu.edu}


\begin{abstract}
	We prove a formula, first obtained by Kleban, Simmons and Ziff using conformal field theory methods, for the (renormalized) density of a critical percolation cluster in the upper half-plane ``anchored'' to a point on the real line. The proof is inspired by the method of images. We also show that more general bulk-boundary connection probabilities have well-defined, scale-covariant scaling limits, and prove a formula for the scaling limit of the (renormalized) density of the critical percolation gasket in any domain conformally equivalent to the unit disk. 
\end{abstract}


\maketitle

\section{Introduction}

Despite being one of the first models for which the emergence of conformal invariance in the scaling limit was verified rigorously, with Smirnov's proof \cite{Smi01} of Cardy's formula \cite{Cardy92} and the subsequent proofs of convergence of interfaces to SLE$_6$ \cite{Smi01,CN07} and CLE$_6$ \cite{CN06,CN08}, critical percolation does not fit as easily as other models into the conformal field theory (CFT) framework, due to the lack of an obvious local field such as the spin (magnetization) field in the Ising model.

In the physics literature, direct reference to local percolation fields is often avoided by deriving results for the $q$-state Potts model and then extrapolating them to percolation by taking the limit $q \to 1$. This is the case, for example, in \cite{Cardy92} and \cite{DV11}. Working with the Potts model, which has a well-defined magnetization field, makes it possible to use CFT tools that are not directly available for percolation. However, the limit $q \to 1$ hides a remarkable amount of subtlety (see, e.g., \cite{CR13}) and, at the moment, has no rigorous justification.

From a mathematical perspective, the rigorous derivation of conformal covariance for observables that can be interpreted as percolation $n$-point correlation functions in the bulk and for a local percolation field has been achieved only very recently \cite{Camia23}.

Given this situation, from a mathematical physics perspective, it is worthwhile to try to match CFT derivations with rigorous results. This note does that by combining the recent results and methods of \cite{Camia23} with insight and ideas from \cite{KSZ06}.

In \cite{KSZ06}, among other results, P. Kleban, J. J. H. Simmons and R. M. Ziff derive a formula for what they call the density of a critical percolation cluster in the upper half-plane ``anchored'' to the origin. Their derivation uses non-rigorous conformal field theory methods and relies on the interpretation of the desired quantity as the two-point function, in the upper half-plane, of two conformal primary fields. Using a method of images, this two-point function in the upper half-plane can be written as a three-point function on the full plane, which is fixed by conformal covariance up to a multiplicative constant. 

In this note, we show how the results and methods of \cite{Camia23} can be used to obtain the formula of Kleban, Simmons and Ziff rigorously as the scaling limit of the probability, appropriately rescaled, of the event that a critical percolation cluster restricted to the upper half-plane and touching the origin of the plane contains a point $z=x+iy$ with $y>0$. The proof of the formula proceeds by considering two ``specular'' copies of the same event, one in the upper half-plane and one in the lower half-plane, and is reminiscent of the method of images used in electrostatics and field theory.


The results and methods of \cite{Camia23} also allow us to show that other bulk-boundary connection probabilities have well-defined, scale-covariant scaling limits, and to derive a formula for the scaling limit of the probability, appropriately rescaled, of the event that a point in a domain belongs to the critical percolation gasket of that domain. (The percolation gasket of a domain with, say, open boundary condition is the set of all open vertices in the domain that belong to the open cluster touching the boundary).


\section{Scaling limit of gaskets and anchored clusters}

We consider critical site percolation on $a\mathbb{T}$, the triangular lattice $\mathbb{T}$ scaled by a factor $a>0$. We embed $\mathbb{T}$ in $\mathbb{R}^2$ as in Figure~\ref{fig-tri-hex} and in such a way that one of its vertices coincides with the origin of $\mathbb{R}^2$. We denote this vertex by $0$ and call it the \emph{origin} of $\mathbb{T}$. With $z^a$ we denote a vertex of $a\mathbb{T}$ as well as, with a slight abuse of notation, the elementary hexagon of the dual hexagonal lattice $a\mathbb{H}$ of which $z^a$ is the center (see Figure \ref{fig-tri-hex}).
Each vertex of $a\mathbb{T}$ (or hexagonal cell of $a\mathbb{H}$) is declared \emph{open} or \emph{closed} with equal probability, independently of all other vertices. We let $\mathbb{P}^a$ denote the corresponding probability measure on configurations of open and closed vertices.

\begin{figure}[!ht]
	\begin{center}
		\includegraphics[width=5cm]{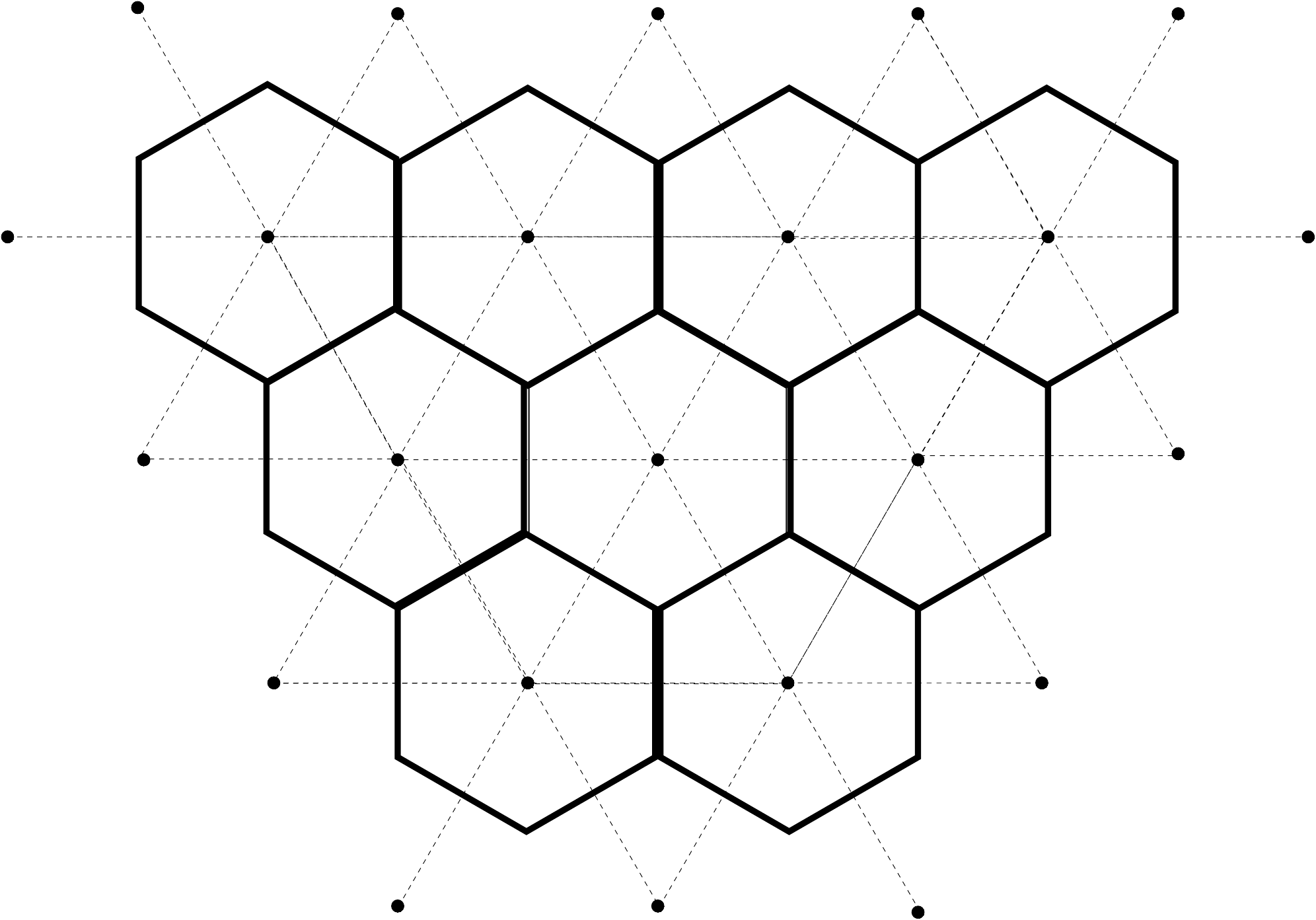}
		\caption{Embedding of the triangular and hexagonal lattices in ${\mathbb R}^2$.}
		\label{fig-tri-hex}
	\end{center}
\end{figure}


Let $\pi_a$ denote the probability of a \emph{bulk one-arm event}, that is, the probability that a vertex $z^a$ belongs to an open cluster that is not fully contained in the disk $B_1(z^a)$ of radius $1$ centered at $z^a$. Analogously, let $\overline\pi_a$ denote the probability of a \emph{boundary one-arm event}, that is, the probability that $0$ belongs to an open cluster whose restriction to the upper half-plane is not fully contained in the disk $B_1(0)$ of radius $1$ centered at $0$.

We write $\pi_a(\varepsilon) = \mathbb{P}^a\big( z^a \longleftrightarrow \partial B_{\varepsilon}(z^a) \big)$, where $z^a \longleftrightarrow \partial B_{\varepsilon}(z^a)$ denotes the event that $z^a$ belongs to an open cluster that is not fully contained in the disk $B_{\varepsilon}(z^a)$ of radius $\varepsilon$ centered at $z^a$. More generally, for a planar domain $D$ containing $z^a$, we let $z^a \longleftrightarrow \partial D$ denote the event that $z^a$ belongs to an open cluster that is not fully contained inside $D$.

Note that $\pi_a=\pi_a(1)$. Moreover, one can prove \cite{GPS13} that, for any $\varepsilon>0$,
\begin{align} \label{eq:scaling-bulk}
\lim_{a \to 0} \pi_a^{-1} \pi_a(\varepsilon) = \lim_{a \to 0} \pi_a^{-1} \mathbb{P}^a\big(z^a \longleftrightarrow \partial B_{\varepsilon}(z^a) \big) = \varepsilon^{-5/48}
\end{align}
and
\begin{align} \label{eq:scaling-boundary}
\lim_{a \to 0} \overline{\pi}_a^{-1} \overline{\pi}_a(\varepsilon) = \varepsilon^{-1/3}.
\end{align}

In the rest of the paper, it is convenient to identify $\mathbb{R}^2$ with the complex plane $\mathbb{C}$.
Our first result concerns the density $g_D$ of the percolation ``gasket'' in a planar domain $D$ (an open subset of $\mathbb{C}$), that is, the scaling limit of the probability that a vertex contained in $D$ belongs to an open cluster that reaches the boundary of $D$.
\begin{theorem} \label{thm:gasket-density}
	Let $D$ be a domain conformally equivalent to the upper half-plane. For any $z \in D$, let $ z^a \in a\mathbb{T}$ be chosen in such a way that $z^a \to z$, as $a \to 0$. Then there exists a constant $C_g \in (0,\infty)$ such that
	\begin{align} \label{eq:gasket-density}
	& g_D(z) := \lim_{a \to 0} \pi_a^{-1} \mathbb{P}^a(z^a \longleftrightarrow \partial D) = C_g \, \text{\emph{rad}}(z,D)^{-5/48},
	\end{align}
	where $\text{\emph{rad}}(z,D)$ denotes the conformal radius of $D$ from $z$.
\end{theorem}

\begin{proof}
	By standard RSW arguments (see, e.g., the proofs of Lemmas 2.1 and 2.2 of \cite{CN09}), there are constants $0<K_1<K_2<\infty$, independent of $a$, such that
	\begin{align}
	K_1 \pi_a \leq \mathbb{P}^a(z^a \longleftrightarrow \partial D) \leq K_2 \pi_a,
	\end{align}
	which shows that $\pi_a^{-1} \mathbb{P}^a(z^a \longleftrightarrow \partial D)$ stays bounded away from zero and infinity as $a \to 0$.
	
	Take $\varepsilon>0$ such that $z$ is at distance greater than $\varepsilon$ from $\partial D$. Then, for all $a$ sufficiently small, $z^a$ is at distance greater than $\varepsilon$ from $\partial D$ and the event $z^a \longleftrightarrow \partial D$ implies the events $z^a \longleftrightarrow \partial B_{\varepsilon}(z^a)$. Therefore, using \eqref{eq:scaling-bulk}, we have
	\begin{align} \label{eq:conditional-probability}
	\begin{split}
	& \lim_{a \to 0} \pi_a^{-1} \mathbb{P}^a\big(z^a \longleftrightarrow \partial D \big) \\
	& \quad = \lim_{a \to 0} \pi_a^{-1} \mathbb{P}^a\big(z^a \longleftrightarrow \partial D \vert z^a \longleftrightarrow \partial B_{\varepsilon}(z^a) \big) \mathbb{P}^a\big(z^a \longleftrightarrow \partial B_{\varepsilon}(z^a) \big) \\
	& \quad = \varepsilon^{-5/48} \lim_{a \to 0} \mathbb{P}^a\big(z^a \longleftrightarrow \partial D \vert z^a \longleftrightarrow \partial B_{\varepsilon}(z^a) \big).
	\end{split}
	\end{align}
	
	The last conditional probability is similar to that appearing in the first line of equation (2.23) of \cite{Camia23}, it is the probability of a similar connectivity event and the conditioning is of the same type. This implies that the arguments used to show the existence of the limit in the first line of (2.23) of \cite{Camia23} can be used to show the existence of the limit of the conditional probability in the last line of \eqref{eq:conditional-probability}, so 
	we can define
	\begin{align}
	\begin{split}
	& \mathbb{P}\big(z \longleftrightarrow \partial D \vert z \longleftrightarrow \partial B_{\varepsilon}(z) \big) := \lim_{a \to 0} \mathbb{P}^a\big(z^a \longleftrightarrow \partial D \vert z^a \longleftrightarrow \partial B_{\varepsilon}(z^a) \big)
	\end{split}
	\end{align}
	and
	\begin{align}
	\begin{split} \label{def:g_D}
	& g_{D}(z) := \lim_{a \to 0} \pi_a^{-1} \mathbb{P}^a\big(z^a \longleftrightarrow \partial D \big) = \varepsilon^{-5/48} \mathbb{P}\big(z \longleftrightarrow \partial D \vert z \longleftrightarrow \partial B_{\varepsilon}(z) \big),
	\end{split}
	\end{align}
	where the last equality is valid for all $\varepsilon>0$ sufficiently small.
	
	The rest of the proof proceeds like the proof of Theorem 1.4 of \cite{Camia23}; we provide a sketch of the argument for completeness. The conformal invariance properties of the scaling limit of percolation \cite{CN06,CN08} imply that, for any domain $D'$ conformally equivalent to $D$, given a conformal map $\phi:D \to D'$, for all $\varepsilon>0$ sufficiently small, 
	\begin{align}
	\begin{split} \label{eq:g_Dprime}
	& g_{D'}(\phi(z)) = \varepsilon^{-5/48} \mathbb{P}\big(\phi(z) \longleftrightarrow \partial D' \vert \phi(z) \longleftrightarrow \partial B_{\varepsilon}(\phi(z)) \big) \\
	& \qquad \qquad = \varepsilon^{-5/48} \mathbb{P}\big(z \longleftrightarrow \partial D \vert z \longleftrightarrow \phi^{-1}\big(\partial B_{\varepsilon}(\phi(z))\big) \big),
	\end{split}
	\end{align}
	where
	\begin{align}
	\begin{split}
	\mathbb{P}\big(z \longleftrightarrow \partial D \vert z \longleftrightarrow \phi^{-1}\big(\partial B_{\varepsilon}(\phi(z))\big) \big) = \lim_{a \to 0} \mathbb{P}^a\big(z^a \longleftrightarrow \partial D \vert \phi^{-1}\big(\partial B_{\varepsilon}(\phi(z))\big) \big).
	\end{split}
	\end{align}
	
	Now let $s = \phi'(z)$ and let $A_{r,R}(z) = B_{R}(z) \setminus B_{r}(z)$ denote the thinnest annulus centered at $z$ containing the symmetric difference of $\phi^{-1}(B_{\varepsilon}(\phi(z)))$ and $B_{\varepsilon/s}(z)$. Since $\phi^{-1}$ is analytic and $(\phi^{-1})'(\phi(z))=1/s$, for every $w \in \partial B_{\varepsilon}(\phi(z))$, $|z-\phi^{-1}(w)| = \varepsilon/s + O(\varepsilon^2)$, which implies that
	\begin{align} \label{eq:limits}
	\lim_{\varepsilon \to 0} \frac{r}{\varepsilon} = \lim_{\varepsilon \to 0} \frac{R}{\varepsilon} = \frac{1}{s}.
	\end{align}
	
	Using \eqref{eq:g_Dprime} and an analog of equation (2.32) of \cite{Camia23}, we have that
	\begin{align}
	\begin{split}
	& \Big(\frac{\varepsilon}{r}\Big)^{-5/48} r^{-5/48} \, \mathbb{P}\big(z \longleftrightarrow \partial D \, \vert \, z \longleftrightarrow \partial B_{r}(z) \big) \\
	& \quad \leq \varepsilon^{-5/48} \mathbb{P}\big(z \longleftrightarrow \partial D \, \vert \, z \longleftrightarrow \partial \phi^{-1}\big(\partial B_{\varepsilon}(z)\big) \big) \\
	& \qquad \leq \Big(\frac{\varepsilon}{R}\Big)^{-5/48} R^{-5/48} \, \mathbb{P}\big(z \longleftrightarrow \partial D \, \vert \, z \longleftrightarrow \partial B_{R}(z) \big).
	\end{split}
	\end{align}
	Combining this with \eqref{def:g_D} and \eqref{eq:g_Dprime} and taking $\varepsilon$ (and hence $r$ and $R$) sufficiently small, we obtain
	\begin{align}
	\begin{split}
	\Big(\frac{\varepsilon}{r}\Big)^{-5/48} \, g_D(z) \leq g_{D'}(\phi(z)) \leq \Big(\frac{\varepsilon}{R}\Big)^{-5/48} \, g_D(z).
	\end{split}
	\end{align}
	Sending $\varepsilon \to 0$ and using \eqref{eq:limits} gives
	\begin{align} \label{eq:conf-cov-g}
	g_{D'}(\phi(z)) = s^{-5/48} g_D(z).
	\end{align} 
	To conclude the proof, apply \eqref{eq:conf-cov-g} to the unit disk, $D'=\mathbb{D}$, choose $\phi: D \to \mathbb{D}$ to be a conformal map such that $\phi(z)=0$, and let $C_g=g_{\mathbb{D}}(0)$.
\end{proof}

The next theorem is the main result of this paper, it provides a rigorous derivation of a formula, first obtained by Kleban, Simmons and Ziff \cite{KSZ06} using conformal field theory methods, for the (renormalized) density of a critical percolation cluster in the upper half-plane ``anchored'' to a point on the real line. We will approximate the upper half-plane $\mathbb{H}$ by the discrete domain depicted in Fig.~\ref{fig-H-lattice}, with origin $0$, which we will refer to as the \emph{discrete upper half-plane}. (In the proof of the theorem, we will also use a similar discrete approximation of the lower half-plane, with origin $0'$, as in Fig.~\ref{fig-H-lattice}.) For any hexagon $z^a$ in the discrete upper half-plane, we will write $z^a \stackrel{\mathbb{H}}{\longleftrightarrow} 0$ to denote the event that $z^a$ and $0$ belong to the same open cluster and, moreover, there is an open path (i.e., a sequence of open, nearest-neighbor hexagons) from $z^a$ to $0$ contained in the discrete upper half-plane. Similarly, $0 \stackrel{\mathbb{H}}{\longleftrightarrow} \partial B_{\varepsilon}(0)$ will denote the event that there is an open path in the discrete upper half-plane from $0$ to the (upper half of the) boundary of the disk $B_{\varepsilon}(0)$.

\begin{figure}[!ht]
	\begin{center}
		\includegraphics[width=5cm]{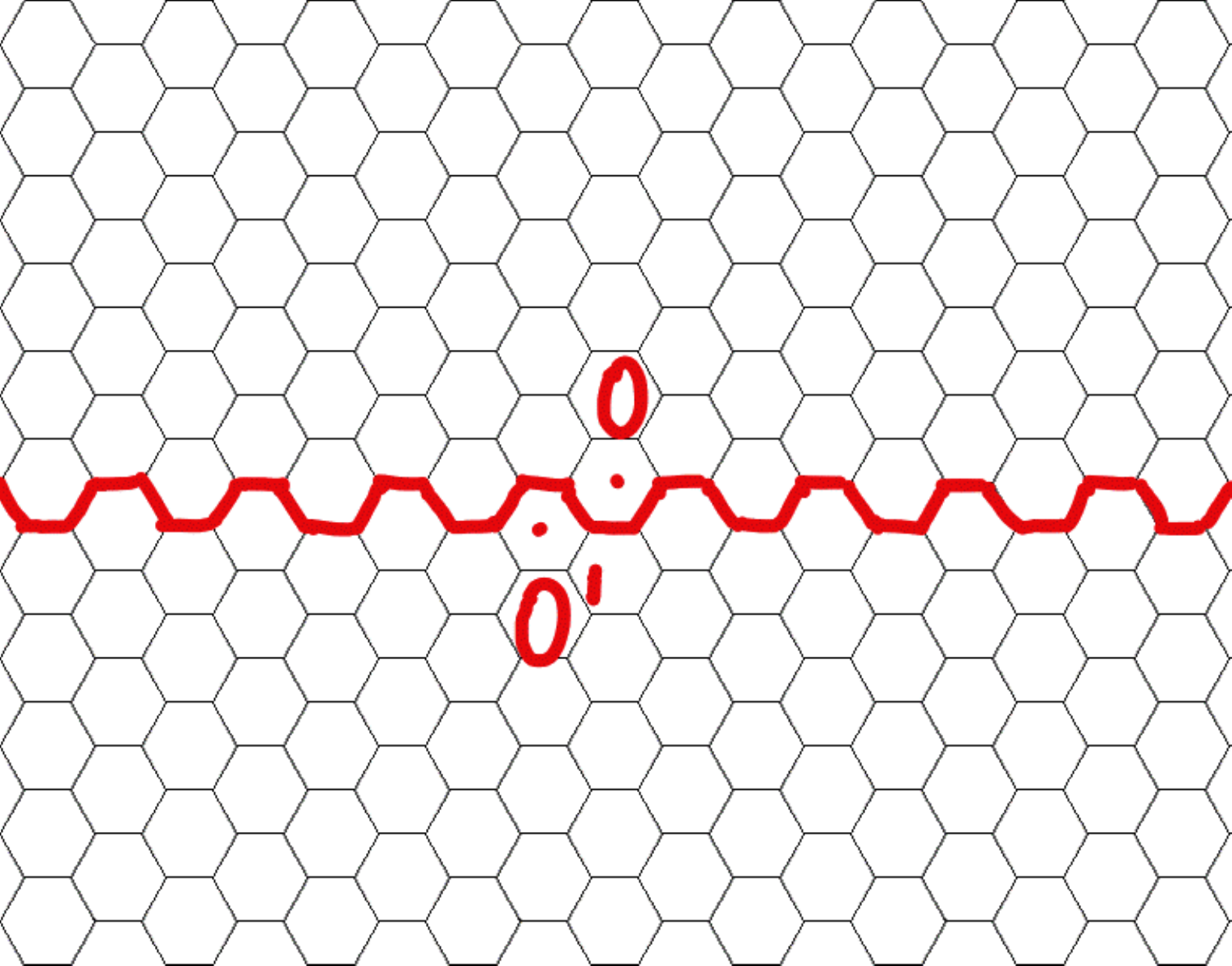}
		\caption{Discrete approximations of the upper and lower half-planes with the respective origins, $O$ and $O'$.}
		\label{fig-H-lattice}
	\end{center}
\end{figure}

\begin{theorem} \label{thm:KSZ-formula}
Consider critical site percolation on the triangular lattice. Given $z=x+iy=r e^{i \theta} \in \mathbb{C}$ with $y>0$, for any sequence of vertices $z^a \in a\mathbb{T}$ in the discrete upper half-plane chosen so that $z^a \to z$ as $a \to 0$, there is a constant $C_{\mathbb{H}} \in (0,\infty)$ such that
\begin{align} \label{eq:KSZ}
\lim_{a \to 0} \pi_a^{-1} \overline{\pi}_a^{-1} \mathbb{P}^a\big(z^a \stackrel{\mathbb{H}}{\longleftrightarrow} 0\big) = C_{\mathbb{H}} \, \frac{y^{11/48}}{\vert z \vert^{2/3}} = C_{\mathbb{H}} \, \frac{(\sin\theta)^{11/48}}{r^{7/16}}.
\end{align}
\end{theorem}

\begin{proof}
By standard RSW arguments (see, e.g., the proofs of Lemmas 2.1 and 2.2 of \cite{CN09}), there are constants $0<K_1<K_2<\infty$, independent of $a$, such that
\begin{align}
K_1 \pi_a \overline{\pi}_a \leq \mathbb{P}^a\big(z^a \stackrel{\mathbb{H}}{\longleftrightarrow} 0\big) \leq K_2 \pi_a \overline{\pi}_a,
\end{align}
which shows that $\pi_a^{-1} \overline{\pi}_a^{-1} \mathbb{P}^a\big(z^a \stackrel{\mathbb{H}}{\longleftrightarrow} 0\big)$ stays bounded away from zero and infinity as $a \to 0$.

Take $\varepsilon>0$ such that $z$ is at distance greater than $2\varepsilon$ from the real line. Then, for all $a$ sufficiently small, $z^a$ is at distance greater than $2\varepsilon$ from the real line and the event $z^a \stackrel{\mathbb{H}}{\longleftrightarrow} 0$ implies the independent events $z^a \longleftrightarrow \partial B_{\varepsilon}(z^a)$ and $0 \stackrel{\mathbb{H}}{\longleftrightarrow} \partial B_{\varepsilon}(0)$. Therefore, using \eqref{eq:scaling-bulk} and \eqref{eq:scaling-boundary}, we have
\begin{align} \label{eq:conditional-prob}
\begin{split}
& \lim_{a \to 0} \pi_a^{-1} \overline{\pi}_a^{-1} \mathbb{P}^a\big(z^a \stackrel{\mathbb{H}}{\longleftrightarrow} 0\big) \\
& \quad = \lim_{a \to 0} \pi_a^{-1} \overline{\pi}_a^{-1} \mathbb{P}^a\big(z^a \stackrel{\mathbb{H}}{\longleftrightarrow} 0 \vert \{z^a \longleftrightarrow \partial B_{\varepsilon}(z^a)\} \cap \{x^a \stackrel{\mathbb{H}}{\longleftrightarrow} \partial B_{\varepsilon}(0)\} \big) \\
& \qquad \quad \mathbb{P}^a\big(z^a \longleftrightarrow \partial B_{\varepsilon}(z^a) \big) \mathbb{P}^a\big(0 \stackrel{\mathbb{H}}{\longleftrightarrow} \partial B_{\varepsilon}(0) \big) \\
& \quad = \varepsilon^{-5/48-1/3} \lim_{a \to 0} \mathbb{P}^a\big(z^a \stackrel{\mathbb{H}}{\longleftrightarrow} 0 \vert \{z^a \longleftrightarrow \partial B_{\varepsilon}(z^a)\} \cap \{0 \stackrel{\mathbb{H}}{\longleftrightarrow} \partial B_{\varepsilon}(0)\} \big).
\end{split}
\end{align}

The last conditional probability is similar to that appearing in the first line of equation (2.23) of \cite{Camia23}, it is the probability of a similar connectivity event and the conditioning is of the same type. This implies that the arguments used to show the existence of the limit in the first line of (2.23) of \cite{Camia23} can be used to show the existence of the limit of the conditional probability in the last line of \eqref{eq:conditional-prob}, so we can define
\begin{align}
\begin{split}
& \mathbb{P}\big(z \stackrel{\mathbb{H}}{\longleftrightarrow} 0 \vert \{z \longleftrightarrow \partial B_{\varepsilon}(z)\} \cap \{0 \stackrel{\mathbb{H}}{\longleftrightarrow} \partial B_{\varepsilon}(0)\} \big) \\
& \quad := \lim_{a \to 0} \mathbb{P}^a\big(z^a \stackrel{\mathbb{H}}{\longleftrightarrow} 0 \vert \{z^a \longleftrightarrow \partial B_{\varepsilon}(z^a)\} \cap \{0 \stackrel{\mathbb{H}}{\longleftrightarrow} \partial B_{\varepsilon}(0)\} \big)
\end{split}
\end{align}
and
\begin{align}
\begin{split}
& \rho_{\mathbb{H}}(z) 
:= \varepsilon^{-7/16} \mathbb{P}\big(z \stackrel{\mathbb{H}}{\longleftrightarrow} 0 \vert \{z \longleftrightarrow \partial B_{\varepsilon}(z)\} \cap \{0 \stackrel{\mathbb{H}}{\longleftrightarrow} \partial B_{\varepsilon}(0)\} \big),
\end{split}
\end{align}
valid for all $\varepsilon>0$ sufficiently small.

Now, keeping in mind Fig.~\ref{fig-H-lattice}, consider the event $\mathcal{A}^a$ that $z^a$ is connected to $0$ by an open path in the discrete upper half-plane (i.e. $z^a \stackrel{\mathbb{H}}{\longleftrightarrow} 0$) \emph{and} $\hat{z}^a$ is connected to $0'$ by an open path in the discrete \emph{lower} half-plane, where $\hat{z}^a$ is the image of $z^a$ under a reflection and a translation that map the discrete upper half-plane to the discrete lower half-plane, mapping $0$ to $0'$.
$\mathcal{A}^a$ can be written as the intersection $\mathcal{A}^a=\mathcal{A}^a_u \cap \mathcal{A}^a_l$ of two \emph{independent} events, $\mathcal{A}^a_u=z^a \stackrel{\mathbb{H}}{\longleftrightarrow} 0$ and a corresponding event $\mathcal{A}^a_l$ in the discrete lower half-plane, so that, by reflection and translation invariance,
\begin{align} \label{eq:equivalence1}
\mathbb{P}^a(\mathcal{A}^a) = \Big(\mathbb{P}^a\big(z^a \stackrel{\mathbb{H}}{\longleftrightarrow} 0\big)\Big)^2.
\end{align}

If we write
\begin{align} \label{eq:limit}
f(0,z,\bar{z}) := \lim_{a \to 0} \pi_a^{-2} \overline{\pi}_a^{-2} \mathbb{P}^a(\mathcal{A}^a),
\end{align}
since $\mathcal{A}^a$ is a connectivity event of the same type as those considered in Theorem 1.1 of \cite{Camia23}, the arguments in the proof of Theorem 1.1 of \cite{Camia23} can be used to show that the limit in \eqref{eq:limit} exists in $(0,\infty)$.

We can define $f(z_0,z_1,z_2)$ for any triple of points $z_0,z_1,z_2$ that form an isosceles triangle, with $z_0$ denoting the vertex opposite the base of the triangle, using \eqref{eq:limit} in the obvious way, i.e., considering the two half-planes bounded by the line through $z_0$ perpendicular to the base of the triangle and placing the origin at $z_0$.
It is clear, then, that $f$ is invariant under translations and rotations.

We now show that $f$ scales covariantly under dilations.
Firstly, note that \eqref{eq:equivalence1} implies
\begin{align} \label{eq:equivalence2}
f(0,z,\bar{z}) = \big(\rho_{\mathbb{H}}(z)\big)^2.
\end{align}
Next, consider $s>0$ and take $\varepsilon>0$ such that both $z$ and $sz$ are at distance greater than $2\varepsilon$ from the real line. Then we have
\begin{align}
\begin{split}
& \rho_{\mathbb{H}}(sz) = \lim_{a \to 0} \pi_a^{-1} \overline{\pi}_a^{-1} \mathbb{P}^a\big(sz^a \longleftrightarrow 0\big) \\
& \quad = \varepsilon^{-7/16} \mathbb{P}\big(sz \stackrel{\mathbb{H}}{\longleftrightarrow} 0 \vert \{sz \longleftrightarrow \partial B_{\varepsilon}(sz)\} \cap \{0 \stackrel{\mathbb{H}}{\longleftrightarrow} \partial B_{\varepsilon}(0)\} \big) \\
& \quad = \varepsilon^{-7/16} \mathbb{P}\big(z \stackrel{\mathbb{H}}{\longleftrightarrow} 0 \vert \{z \longleftrightarrow \partial B_{\varepsilon/s}(z)\} \cap \{0 \stackrel{\mathbb{H}}{\longleftrightarrow} \partial B_{\varepsilon/s}(0)\} \big),
\end{split}
\end{align}
where the last equality follows from the scale invariance of the scaling limit of critical percolation \cite{CN06,CN08}.
This implies that
\begin{align}
\begin{split} \label{eq:scale-invariance}
& \rho_{\mathbb{H}}(sz) = s^{-7/16} (\varepsilon/s)^{-7/16} \mathbb{P}\big(z \stackrel{\mathbb{H}}{\longleftrightarrow} 0 \vert \{z \longleftrightarrow \partial B_{\varepsilon/s}(z)\} \cap \{0 \stackrel{\mathbb{H}}{\longleftrightarrow} \partial B_{\varepsilon/s}(0)\} \big) \\
& \qquad \quad = s^{-7/16} \rho_{\mathbb{H}}(z),
\end{split}
\end{align}
which proves the covariance of $\rho_{\mathbb{H}}$, and therefore of $f$ (by \eqref{eq:equivalence2}), under scale transformations and is consistent with \eqref{eq:KSZ}.

Given the covariance properties of $f$, standard considerations (see, e.g., Section 4.3.1 of \cite{DFMS} or the proof of Theorem 4.5 of \cite{CGK16}) 
imply that 
\begin{align} \label{eq:functional-form}
f(0,z,\bar{z}) = \frac{C}{\vert z \vert^{a} \vert \bar{z} \vert^{b} \vert z-\bar{z} \vert^{c}},
\end{align}
for some constant $C \in (0,\infty)$ and with $a = b = 5/48 + 2/3 - 5/48 = 2/3$ and $c = 5/48 + 5/48 - 2/3 = -11/24$.
Therefore, combining \eqref{eq:functional-form} and \eqref{eq:equivalence2}, we have that
\begin{align}
\rho_{\mathbb{H}}(z) = \sqrt{f(0,z,\bar{z})} = \sqrt{C} \vert z-\bar{z} \vert^{11/48} \vert z \vert^{-1/3} \vert \bar{z} \vert^{-1/3}.
\end{align}
%
%
Writing $z=x+iy=r e^{i \theta}$, this gives
\begin{align}
\rho_{\mathbb{H}}(z) = C_{\mathbb{H}} \, y^{11/48} \, \vert z \vert^{-2/3} = C_{\mathbb{H}} \, r^{-7/16} \, (\sin\theta)^{11/48},
\end{align}
where $C_{\mathbb{H}}= 2^{11/48} \sqrt{C}$.
\end{proof}

We conclude this paper with an observation. The first part of the proof of Theorem \ref{thm:KSZ-formula} (up to equation \eqref{eq:scale-invariance}) applies to more general ``bulk-boundary'' connection probabilities, involving multiple points in the upper half-plane and on the real line. In general, there is no known explicit formula, but one can still prove the existence of the limit and its scale covariance.
\begin{theorem} Consider critical site percolation on the triangular lattice. Given $z_1,\ldots,z_k$ in the upper half-plane and $x_1,\ldots,x_n \in \mathbb{R}$, for any sequences of vertices $z_1^a, \ldots, z_k^a$ and $x_1^a \ldots, x^a_n$ of $a\mathbb{T}$ in the discrete upper half-plane with $z^a_1 \to z_1, \ldots, z^a_k \to z_k$ and $x^a_1 \to x_1, \ldots, x^a_n \to x_n$ as $a \to 0$, the limit
\begin{align}
\lim_{a \to 0} \pi_a^{-k} \overline{\pi}_a^{-n} \mathbb{P}^a\big(z^a_i \longleftrightarrow x^a_j \; \forall i,j \big) =: \rho_{\mathbb{H}}(z_1, \ldots, z_k;x_1, \ldots, x_n)
\end{align}
exists in $(0,\infty)$. Moreover, for any $s>0$,
\begin{align} \label{eq:scale-covariance}
\rho_{\mathbb{H}}(sz_1, \ldots, sz_k;sx_1,\ldots,sx_n) = s^{-5k/48-n/3} \rho_{\mathbb{H}}(z_1, \ldots, z_k;x_1,\ldots,x_n).
\end{align}
\end{theorem}

\begin{proof}[Idea of the proof.] The existence of the limit follows from the same arguments used in the proof of Theorem 1.1 of \cite{Camia23} and of the limit in \eqref{eq:conditional-prob}, with the difference that now one needs to condition on $k+n$ one-arm events in the upper half-plane, at $z^a_1, \dots, z^a_k$ and $x^a_1, \ldots, x^a_n$. The scale covariance is proved in the same way as \eqref{eq:scale-invariance}, with the exponent coming from the exponents in \eqref{eq:scaling-bulk} and \eqref{eq:scaling-boundary}.
\end{proof}

\medskip

\noindent {\bf Acknowledgments:} The author thanks Peter Kleban and Robert Ziff for an interesting correspondence and for comments on a draft of the paper.

\medskip





\end{document}